\documentclass[11pt]{article}
\usepackage[authoryear]{natbib}
\usepackage{amssymb}
\usepackage{amsfonts}
\usepackage{amsmath}
\usepackage{dsfont}
\usepackage{soul}
\usepackage[nohead]{geometry}
\usepackage{graphicx}
\usepackage{amsthm}
\usepackage{color}
\usepackage{comment}
\usepackage{setspace}
\usepackage{framed}
\usepackage{enumitem}
\usepackage{todonotes}
\usepackage{tikz}
\usepackage{booktabs}
\usepackage{hyperref}
\usepackage{rotating}
\usepackage{subfigure}
\usepackage[flushleft]{threeparttable}
\usepackage{multirow}
\newcommand{\cref}[2][1]{{\textup{(\hyperref[#2]{\ref*{#2}$_{#1}$})}}}
\usepackage{xr}
\usepackage{bm}
\allowdisplaybreaks

7

\newcommand{\diag}{\mathrm{diag}}

\newcommand{\var}{\mathrm{var}}
\newcommand{\beq}{\begin{eqnarray*}}
	\newcommand{\eeq}{\end{eqnarray*}}
\newtheorem{thm}{Theorem}[section]

\newtheorem{lem}{Lemma}[section]

\newtheorem{assum}{Assumption}[section]
\newtheorem{pro}{Proposition}[section]
\numberwithin{equation}{section}
\theoremstyle{definition}

\newtheorem{remark}{Remark}[section]
\makeatletter
\def\@biblabel#1{\hspace*{-\labelsep}}
\makeatother
\begin{document}
	
	\title{Feasible Generalized Least Squares for Panel Data with Cross-sectional and Serial Correlations}
	\date{\today }

	\author{
		Jushan Bai\thanks{%
			Address: 420 West 118th St. MC 3308, New York, NY 10027, USA. E-mail: \texttt{ jb3064@columbia.edu}.} \\ \footnotesize Columbia University 
		\and Sung Hoon Choi\thanks{%
			Address: 75 Hamilton St., New Brunswick, NJ 08901, USA. E-mail: 
			\texttt{sc1711@economics.rutgers.edu}.} \\  \footnotesize  Rutgers University  \and
		Yuan Liao\thanks{Address: 75 Hamilton St., New Brunswick, NJ 08901, USA. Email:
			\texttt{yuan.liao@rutgers.edu}.}\\  \footnotesize   Rutgers University 
	}

	\maketitle
	
	\begin{abstract}
	
		This paper considers generalized least squares (GLS) estimation for linear panel data models. By estimating the large error covariance matrix consistently, the proposed feasible GLS (FGLS) estimator is more efficient than the ordinary least squares (OLS) in the presence of heteroskedasticity, serial and cross-sectional correlations. To take into account the serial correlations, we employ the banding method. To take into account the cross-sectional correlations,  we suggest to use the thresholding method.	We establish the limiting distribution of the proposed estimator. A Monte Carlo study is considered. The proposed method is applied to an empirical application.
		
		Keywords: Panel data, efficiency, thresholding, banding, cross-sectional correlation, serial correlation, heteroskedasticity 
	\end{abstract}
	
	\thispagestyle{empty}
	
	

	\onehalfspacing

	\newpage
	\setcounter{page}{1}
	\pagenumbering{arabic}
	
	\section{Introduction}
	Heteroskedasticity, cross-sectional and serial correlations are important problems in the error terms of panel regression models. There are two approaches to deal with these problems. The first approach is to use the ordinary least squares (OLS) estimator but with a robust standard error that is robust to heteroskedasticity and correlations, for example, \cite{white1980heteroskedasticity}; \cite{newey1986simple}; \cite{liang1986}; \cite{arellano1987}; \cite{driscoll1998}; \cite{hansen2007}; \cite{vogelsang2012}, among others. A widely used class of robust standard errors are clustered standard errors, for example, \cite{petersen2009}, \cite{wooldridge2010econometric} and \cite{cameron2015}.  \cite{bai2019olsse} proposed a robust standard error with unknown clusters. In an interesting paper by  \cite{abadie2017should}, they argued for caution in the application of clustered standard errors since they may give rise to conservative confidence intervals. The second approach is to use the generalized least squares estimator (GLS) that directly takes into account heteroskedasticity, and  cross-sectional and serial correlations in the estimation. It is well known that GLS is more efficient than OLS.
	
	This paper focuses on the second approach.  For panel models, the underlying covariance matrix involves a large number of parameters. It is important to make GLS operational.  We thus consider feasible generalized least squares (FGLS). \cite{hansen2007fgls} studied FGLS estimation that takes into account serial correlation and clustering problems in fixed effects panel and multilevel models. His approach  requires the cluster structure to be known.  This gives motivation to our paper. We assume the unknown cluster structure, and control heteroskedasticity, both serial and cross-sectional correlations by estimating the large error covariance matrix consistently. In cross-sectional setting, \cite{romano2017resurrecting} obtained asymptotically valid inference of the FGLS estimator, combined with heteroskedasticity-consistent standard errors without knowledge of the conditional heteroskedasticity functional form. Moreover, \cite{miller2018feasible} adapted machine learning methods (i.e., support vector regression) to take into account the misspecified form of heteroskedasticity.

	In this paper, we consider (i) balanced panel data, (ii) the case of  large-$N$ large-$T$, and  (iii) both serial and cross-sectional correlations, but unknown stucture of clusters.
	We introduce a modified FGLS estimator that eliminates the cross-sectional and serial correlation bias by proposing a high-dimensional  error covariance matrix estimator. In addition, our proposed method is applicable when the knowledge of clusters is not available.  Following an idea suggested in \cite{bailaio2017}, in this paper, the FGLS involves estimating an $NT\times NT$ dimensional  inverse covariance matrix $\Omega^{-1}$, where 
	$$
	\Omega=(Eu_{t}u_{s}')
	$$
	where each block $Eu_{t}u_{s}'$ is an $N\times N$ autocovariance matrix. Here parametric structures on the serial or cross-sectional correlations are not imposed. By assuming weak dependences, we apply nonparametric methods to estimate the covariance matrix. To control the autocorrelation in time series, we employ the idea of Newey-West truncation. This method, in the FGLS setting, is 
	equivalent to  ``banding", previously proposed by  
	   \cite{bickel2008b} for estimating large covariance matrices. We apply it to banding out   off-diagonal $N\times N$ blocks that are far from the diagonal block.  
	    In addition, to control for the cross-sectional correlation, we assume that  each of the $N\times N$ block matrices are sparse, potentially resulting from the presence of cross-sectional correlations within clusters.  We then estimate them by applying the thresholding approach of \cite{bickel2008a}. We apply thresholding separately to the $N\times N$ blocks, which are formed by  time lags
	    $Eu_{t}u_{t-h}': h=0,1,2,...$. 
This allows the cluster-membership to be potentially changing over-time.
A contribution of this paper is the theoretical justification for estimating the large error covariance matrix.

	For the FGLS, it is crucial for the asymptotic analysis to prove that the effect of estimating $\Omega$ is first-order negligible. In the usual low-dimensional settings that involve estimating optimal weight matrix, such as the optimal GMM estimations, it has been well known that   consistency for the   inverse covariance matrix estimator is sufficient for the first-order asymptotic theory, e.g., \cite{hansen1982}, \cite{newey1990efficient}, \cite{newey1994large}. However, it turns out that when the covariance matrix is of high-dimensions, not even the optimal convergence rate for estimating  $  \Omega^{-1}$ is sufficient. 
	In fact, proving the first-order equivalence between the FGLS and the infeasible GLS (that uses the true $\Omega^{-1}$) is   a very challenging problem under the large $N$, large $T$ setting. We provide a new theoretical argument to achieve this goal.   
	
	The banding and thresholding methods, which we employ in this paper, are two of the useful regularization methods. In the recent machine learning literature, these methods have been extensively exploited for estimating high-dimensional parameters.  Moreover, in the econometric literature, nonparametric machine learning techniques have been verified to be powerful tools: \cite{baing2017}; \cite{chernozhukov2016, chernozhukov2017}; \cite{wager2018estimation}, etc.
	
	The rest of the paper is organized as follows. In Section \ref{sec:gls}, we describe the model and the large error covariance matrix estimator. Also we introduce the implementation of FGLS estimatior and its limiting distribution. Section \ref{simulation} presents  Monte Carlo studies evaluating the finite sample performance of the estimators. In Section \ref{application}, we apply our methods to study the US divorce rate problem. Conclusions are provided in Section \ref{conclusion}. All  proofs are given in Appendix \ref{appendix}.
	
	Throughout this paper, let $\nu_{\min}(A)$ and $\nu_{\max}(A)$ denote the minimum and maximum eigenvalues of matrix $A$ respectively. Also we use $\|A\| = \sqrt{\nu_{\max}(A'A)}$, $\|A\|_{1} = \max_{i}\sum_{j}|A_{ij}|$ and $\|A\|_{F} = \sqrt{tr(A'A)}$ as the operator norm, $\ell_1$-norm and the Frobenius norm of a matrix A, respectively. Note that if $A$ is a vector, $\|A\| =\|A\|_{F}$ is equal to the Euclidean norm.
	
	\section{Feasible Generalized Least Squares} \label{sec:gls}

	We consider a linear model \footnote{For technical simplicity we focus on a simple model where there are no fixed effects. It is straightforward to allow additive fixed effects $\alpha_i+\mu_t$ by applying the de-meaning first. The theories would be slightly more sophisticated, though such extensions are straightforward.}
	\begin{equation} \label{e.2.1}
	y_{it} = x_{it}'\beta + u_{it}. 	\end{equation}
	The model (\ref{e.2.1}) can be stacked and represented in full matrix notation as
	\begin{equation} \label{e.3.1}
	Y = X\beta + U,  
	\end{equation}
	where  $Y = (y_1',\cdots, y_T')'$ is the $NT \times 1$ vector of $y_{it}$ with each $y_t$ being an $N\times 1$ vector; $X = (x_1',\cdots, x_T')'$ is the $NT \times d$ matrix of $x_{it}$ with each $x_t$ being an $N\times d$; $U = (u_1',\cdots, u_T')'$ is the $NT \times 1$ vector of $u_{it}$ with each $u_t$ being an $N\times 1$ vector. 
	
	Let $\Omega = (Eu_{t}u_{s}')$ be an $NT \times NT$ matrix, consisting of many ‘‘blocks’’ matricies. The $(t,s)$th block is  an $N \times N$ covariance matrix $Eu_tu_s'$.  We consider the following   (infeasible) GLS estimator of $\beta$:
	\begin{equation} \label{e.3.2}
	\widetilde{\beta}_{GLS}^{inf} = (X'\Omega^{-1}X)^{-1}X'\Omega^{-1}Y.
	\end{equation}
	Note that $\Omega$ is a   high-dimensional conditional covariance matrix, which is very difficult to estimate.  
	We aim  to achieve the following: (i) obtain a ``good" estimator of $\Omega^{-1}$, allowing an arbitrary form of weak dependence in $u_{it}$, and (ii) show that the effect of replacing $\Omega^{-1}$ by $\widehat{\Omega}^{-1}$ is asymptotically negligible.
	
	We start with a population approximation for $\Omega$ in order to gain the intuitions. Then, we suggest the estimator for $\Omega$ that takes into account   both correlations problem.
	
	\subsection{Population approximation}
	We start with a ``banding" approximation to control serial correlations. Recall that $\Omega = (Eu_{t}u_{s}')$, where the $(t, s)$ block is $Eu_{t}u_{s}'$. By assuming serial stationarity and strong mixing condition, $Eu_{t}u_{s}'$ depends on $(t, s)$ only through $h=t-s$. Specifically, with slight abuse of notation, we can write $\Omega_{t,s} = \Omega_{h} = Eu_{t}u_{t-h}'$.  Note for $i\neq j$, it is possible that  $Eu_{it}u_{j,t-h}\neq Eu_{i,t-h}u_{jt}$, so $\Omega_{h}$ is possibly  non-symmetric  for $h>0$. On the other hand, $\Omega$ is symmetric due to $\Omega_{s,t} = \Omega_{t,s}'$. The diagonal blocks are the same, and all equal $\Omega_{0}=Eu_{t}u_{t}'$, while magnitudes of the elements of the off-diagonal blocks $\Omega_{h}=Eu_{t}u_{t-h}'$ decay to  zero as $|h|\to\infty$ under the weak  serial dependence assumption. 
	
	In the Newey-West spirit, $\Omega$ can be approximated by $\Omega^{NW} = (\Omega_{t,s}^{NW})$, where each block can be written as $\Omega_{t,s}^{NW}= \Omega_{h}^{NW}$ for $h=t-s$. Here  $\Omega_{h}^{NW}$ is an  $N \times N$ block matrix, defined as:
	\begin{equation*} \label{e.3.3}
	\Omega_{h}^{NW} = \begin{cases}Eu_{t}u_{t-h}', & \text{if} \;\; |h| \leq L\\
	0, & \text{if} \;\; |h| > L,
	\end{cases} 
	\end{equation*}
	for some pre-determined $L \rightarrow \infty$. For instance, as suggested by \cite{newey1994}, we can set $L$ equal to $4(T/100)^{(2/9)}$. Note that $\Omega_{h}^{NW} =\Omega_{-h}^{NW'}$.
	We regard $\Omega^{NW} = (\Omega_{h}^{NW})$ as the ``population banding approximation''.
	
	Next, we focus on the $N\times N$ block matrix $\Omega_{h}=Eu_{t}u_{t-h}' $ to control cross-sectional correlations. Under the intuition that $u_{it}$ is cross-sectional weakly dependent, we assume $\Omega_{h}$ is a sparse matrix, that is, $\Omega_{h,ij}=Eu_{it}u_{j,t-h}$ is ``small" for ``many" pairs $(i,j)$. Then $\Omega_{h}$ can be approximated by a sparse matrix $\Omega_{h}^{BL} = (\Omega_{h,ij}^{BL})_{N \times N}$  (\cite{bickel2008a}), where
	\begin{equation*} \label{e.3.4}
	\Omega_{h,ij}^{BL} =  \begin{cases}Eu_{it}u_{j,t-h}, & \text{if} \;\; |Eu_{it}u_{j,t-h}| >\tau_{ij}\\~ 
	0, & \text{if} \;\; |Eu_{it}u_{j,t-h}| \leq \tau_{ij},
	\end{cases}   
	\end{equation*}
	for some pre-determined threshold $\tau_{ij} \rightarrow 0$. We regard $\Omega_{h}^{BL}$ as the ``population sparse approximation''. 
	
In summary, we approximate $\Omega$ by an $NT\times NT$ matrix 
  $(\widetilde \Omega^{NT}_{t,s})$, where  each block $\widetilde \Omega^{NT}_{t,s}$ is an $N\times N$ matrix, defined as: for $h=t-s$,
$$
\widetilde \Omega^{NT}_{t,s}:= \begin{cases}
	\Omega_h^{BL} , & \text{if} \;\; |h| \leq L  \\
	0, & \text{if} \;\; |h| > L.
	\end{cases} \quad
$$
Therefore, we use ``banding'' to control the serial correlation, and ``sparsity'' to control the cross-sectional correlation.
Note that an advantage of the method proposed in this paper is that it does not assume known cluster information (i.e., the number of clusters and the membership of clusters). Moreover, this method could also be modified to take into account the clustering information when available.

	\subsection{Implementation of Feasible GLS} \label{implemenation}
	\subsubsection{The estimator of $\Omega$ and FGLS} \label{estimators}
	Given the intuition of the  population approximation, we construct the large covariance estimator as follows. First, we denote the OLS estimator of $\beta$ by $\widehat{\beta}_{OLS}$ and the corresponding residuals by $\widehat{u}_{it} = y_{it} - x_{it}'\widehat{\beta}_{OLS}$. 
	
	Now we estimate the $N\times N$ block matrix $\Omega_{h}=Eu_{t}u_{t-h}'$. To do so, let 
	$$
	\widetilde{R}_{h,ij} =\begin{cases}
\frac{1}{T}\sum_{t=h+1}^{T}\widehat{u}_{it}\widehat{u}_{j,t-h}, & \text{ if } h\geq 0\\
\frac{1}{T}\sum_{t=1}^{T+h}\widehat{u}_{it}\widehat{u}_{j,t-h}, & \text{ if } h<0
	\end{cases},\quad \text{ and } \widetilde{\sigma}_{h,ij} =  \begin{cases}
	\widetilde{R}_{h,ii}, & \text{if} \;\; i=j \\
	s_{ij}(\widetilde{R}_{h,ij}), & \text{if} \;\; i \neq j,
	\end{cases}
	$$  
	where $s_{ij}(\cdot) : \mathbb{R} \rightarrow \mathbb{R}$ is a ``soft-thresholding function" with an entry dependent threshold $\tau_{ij}$ such that 
	\begin{equation*} \label{e.3.6}
	s_{ij}(z) = \text{sgn}(z)(|z| - \tau_{ij})_{+},\\
	\end{equation*}
	where $(x)_{+} = x$ if $x\geq 0$, and zero otherwise. Here $\text{sgn}(\cdot)$  denotes the sign function, and other thresholding functions, e.g., hard thresholding, are possible.
	For the threshold value, we specify 
	\begin{equation*} \label{e.3.7}
	\tau_{ij} = M\gamma_{T}\sqrt{|\widetilde{R}_{0,ii}|\;|\widetilde{R}_{0,jj}|},
	\end{equation*}
	for some pre-determined value $M>0$, where $\gamma_{T} = \sqrt{\frac{\log(LN)}{T}}$ is such that $\max_{h \leq L}\max_{i,j \leq N}|\widetilde{R}_{h,ij}-Eu_{it}u_{i,t-h}| = O_{P}(\gamma_{T})$. Note that the constant thresholding parameter could be allowed as \cite{bickel2008a}. In practice, however, it is more desirable to have entry dependent threshold, $\tau_{ij}$. $M$ can be chosen by multifold cross-validation, which is explained in Section \ref{tunning}. 
	Then define
	\begin{equation} \label{e.3.8}
	\widetilde{\Omega}_{h} = (\widetilde{\sigma}_{h,ij})_{N \times N}.
	\end{equation}
	Next, we define the   $(t, s)$th  block $\widehat{\Omega}_{t,s}$ as an $N \times N$ matrix: for $h=t-s$,
	\begin{equation*}  \label{e.3.9}
	\widehat{\Omega}_{t,s} =  \begin{cases}
	\omega(|h|,L)\widetilde{\Omega}_{h}, & \text{if} \;\; |h| \leq L\\
	0, & \text{if} \;\; |h| > L.
	\end{cases}.    
	\end{equation*}
	Here $\omega(h,L)$ is the kernel function (see \cite{andrews1991} and  \cite{newey1994}). We let $\omega(h,L) = 1-h/(L+1)$ be the Bartlett kernel function, where $L$ is the bandwidth. In addition, the choice of $L$ is detailed in Section \ref{tunning}. Our final estimator of $\Omega$ is an $NT\times NT$ matrix:
	$$
	\widehat{\Omega}=(\widehat{\Omega}_{t,s}).
	$$ Here $\widehat{\Omega}$ is a nonparametric estimator, which does not require an assumed parametric structure on $\Omega$. Note that, for the large sample size, the proposed estimator may require a huge computational cost due to use of an $NT \times NT$ matrix. 
	
	Finally, given $\widehat{\Omega}$, we propose the feasible GLS (FGLS) estimator of $\beta$ as
	\begin{equation*} 
	\widehat{\beta}_{FGLS} = [X'\widehat{\Omega}^{-1}X]^{-1}X'\widehat{\Omega}^{-1}Y. 
	\end{equation*}

	\begin{remark}[Universal thresholding]
	We apply thresholding separately to  the $N\times N$ blocks, $ (\widetilde{\sigma}_{h,ij})_{N \times N}$, which are  estimated lagged blocks for $Eu_{t}u_{t-h}: h=0,1,2,...$. 	This allows the cluster-membership to be potentially changing over-time, that is, the identities of zeros and nonzero elements of   $Eu_{t}u_{t-h}$ can change over $h$.   If it is known that the  cluster-membership (i.e., identities  of nonzero elements) is time-invariant, then  one  would   set $\widetilde\sigma_{h,ij}=0$ if $\max_{h\leq L}|\widetilde R_{h,ij}|\leq \tau_{ij}$ for $i\neq j$.  This potentially would increase the finite sample accuracy of identifying the cluster-membership. 
	\end{remark}

	\subsubsection{Choice of tuning parameters} \label{tunning}
	Our suggested covariance matrix estimator, $\widehat{\Omega}$, requires the choice of tuning parameters $L$ and $M$, which are the bandwidth and the threshold constant respectively. We write $\widehat{\Omega}(M,L)=\widehat{\Omega}$, where the covariance estimator depends on $M$ and $L$. First, to choose the bandwidth $L$, we suggest using $L^{*} = 4(T/100)^{2/9}$, which is proposed by \cite{newey1994}. For a small size of $T$, we also recommend $L \leq 3$.

The thresholding constant, $M$, can be chosen through multifold cross-validation. 
We  randomly split the data $P$ times. We divide the data into $P=\log(T)$ blocks $J_1,...,J_P$ with block length $T/\log(T)$ and take one of the $P$ blocks as  the validation set.  	 At the $p$th split, we denote by $\widetilde{\Omega}_{0}^{p}$ the sample covariance matrix based on the validation set, defined by
	$\widetilde{\Omega}_{0}^{p} =  |J_{p}|^{-1}\sum_{t \in J_{p}}\widehat{u}_{t}\widehat{u}_{t}'$.
	Let $\widetilde{\Omega}_{0}^{S,p}(M)$ be the  thresholding estimator with threshold constant $M$ using the training data set $\{\widehat{u}_{t}\}_{t \notin J_{p}}$. 
	Finally, we choose the constant $M^{*}$ by minimizing the cross-validation objective function
	
	\begin{equation*}
	M^* = \arg \min_{c<M<\bar C}\frac{1}{P}\sum_{j=1}^{P}\|\widetilde{\Omega}_{0}^{S,p}(M)-\widetilde{\Omega}_{0}^{p}\|_{F}^2,
	\end{equation*}
	where $\bar C$ is a large constant such that $\widetilde{\Omega}_{0}^{S}(\bar C)$ is a diagonal matrix, and  $c$ is a  constant that guarantees the positive definiteness of $\widehat{\Omega}(M,L)$ for $M>c$: 
	for each fixed $L$,
	$$
	c= \inf[M>0: \lambda_{\min}\{\widehat{\Omega}(C,L)\}>0, \forall C>M].
	$$
	Here $\widetilde{\Omega}_{0}^{S}(M)$ is the soft-thresholded estimator as defined in the equation (\ref{e.3.8}). Then the resulting estimator of $\Omega$ is $\widehat \Omega(M^*,L^*)$.

	\subsection{The effect of $\widehat{\Omega}^{-1}-\Omega^{-1}$} \label{sec3.2}
	
	A key step of proving the asymptotic property for $\widehat{\beta}_{FGLS}$ is to show that it is asymptotically equivalent to $\widetilde{\beta}_{GLS}^{inf}$, that is:
	\begin{equation} \label{e.3.11}
	 \frac{1}{\sqrt{NT}}X'(\widehat{\Omega}^{-1}-\Omega^{-1})U = o_{P}(1).
	\end{equation} 
	In the usual low-dimensional settings that involve estimating optimal weight matrix, such as the optimal GMM estimations, it has been well known that   consistency for the   inverse covariance matrix estimator is sufficient for the first-order asymptotic theory, e.g., \cite{hansen1982}, \cite{newey1990efficient}, \cite{newey1994large}. It turns out, when the covariance matrix is of high-dimensions, not even the optimal convergence rate of $\|\widehat{\Omega} - \Omega\|$ is sufficient. In fact, proving equation (\ref{e.3.11}) is a very challenging problem.  	In the general case when both cross-sectional and serial correlations are present, our strategy is to use a careful expansion for $\frac{1}{\sqrt{NT}}X'(\widehat{\Omega}^{-1}-\Omega^{-1})U$. We shall proceed in two steps:\\~
	Step 1: Show that $\frac{1}{\sqrt{NT}}X'(\widehat{\Omega}^{-1}-\Omega^{-1})U = \frac{1}{\sqrt{NT}}W'(\widehat{\Omega}-\Omega)\varepsilon + o_{P}(1),$ where $W = \Omega^{-1}X$, and $\varepsilon = \Omega^{-1}U$.\\~
	Step 2: Show that $\frac{1}{\sqrt{NT}}W'(\widehat{\Omega}-\Omega)\varepsilon = o_{P}(1)$.\\~
	
	Now we suppose $\omega(h,L) =1, \Omega \approx \Omega^{NW}$ and let $A_{b_h} = \{(i,j) : |Eu_{it}u_{j,t-h}| \neq 0\},  A_{s_h} = \{(i,j) : |Eu_{it}u_{j,t-h}| = 0\}$. As for  Step 2, we shall show,  	
	\begin{equation} \label{e.3.16}
	\frac{1}{\sqrt{NT}}W'(\widehat{\Omega}-\Omega)\varepsilon \thickapprox \frac{1}{\sqrt{NT}}\sum_{|h| \leq L}\sum_{i,j \in A_{b_h}}\sum_{t=h+1}^{T}w_{it}\varepsilon_{j,t-h}\frac{1}{T}\sum_{s=h+1}^{T}(u_{is}u_{j,s-h}-Eu_{it}u_{j,t-h}).
	\end{equation}
	Here $w_{it}$ is defined such that, we can write $W = (w_1',\cdots, w_T')'$ with  $w_t$ being an $N\times d$ matrix of  $w_{it}$; $\varepsilon_{it}$ is defined similarly. 
	We then further argue that the right hand side of (\ref{e.3.16}) is $o_{P}(1)$ by applying a high-level Assumption \ref{assum4}, which essentially saying the right hand side of (\ref{e.3.16}) is $o_P(1)$.
	
	To appreciate the need of this high-level condition, let us consider a simple example as follows.

	\textbf{A simple example.} To illustrate the key technical issue,  consider  a simple and ideal case where $u_{it}$ is known, and independent across both $i$ and $t$, but with cross-sectional heteroskedasticity. In this case, the covariance matrix of the $NT \times 1$ vector $U$ is a diagonal matrix, with diagonal elements $\sigma_{i}^2 = Eu_{it}^2$:
	$$
	\Omega = 
	\begin{pmatrix} 
	D & && \\
	& D&&\\
	&&\ddots&\\
	&&&D 
	\end{pmatrix},
	\text{ where }
	D = 
	\begin{pmatrix} 
	\sigma_{1}^2 & && \\
	& \sigma_{2}^2&&\\
	&&\ddots&\\
	&&&\sigma_{N}^2 
	\end{pmatrix}.
	$$
	Then a natural estimator for $\Omega$ is
	$$
	\widehat{\Omega} = 
	\begin{pmatrix} 
	\widehat{D} & && \\
	& \widehat{D}&&\\
	&&\ddots&\\
	&&&\widehat{D} 
	\end{pmatrix},
	\text{ where }
	\widehat{D} = 
	\begin{pmatrix} 
	\widehat{\sigma}_{1}^2 & && \\
	& \widehat{\sigma}_{2}^2&&\\
	&&\ddots&\\
	&&&\widehat{\sigma}_{N}^2 
	\end{pmatrix},
	$$
	and $\widehat{\sigma}_{i}^2=\frac{1}{T}\sum_{t=1}^{T}u_{it}^2$, because $u_{it}$ is known. Then the  GLS becomes:
	\begin{equation*} \label{e.3.12}
	(\frac{1}{NT}\sum_{i=1}^{N}\sum_{t=1}^{T}x_{it}x_{it}'\widehat{\sigma}_{i}^{-2})^{-1}\frac{1}{NT}\sum_{i=1}^{N}\sum_{t=1}^{T}x_{it}y_{it}\widehat{\sigma}_{i}^{-2}.
	\end{equation*}
	A key step is to prove that the effect of estimating $D$ is asymptotically negligible: 
	\begin{equation*} \label{e.3.13}
	\frac{1}{\sqrt{NT}}\sum_{i=1}^{N}\sum_{t=1}^{T}x_{it}u_{it}(\widehat{\sigma}_{i}^{-2}-\sigma_{i}^{-2}) = o_{P}(1).
	\end{equation*}
	It can be shown that the problem reduces to proving: $$A \equiv \frac{1}{\sqrt{NT}}\sum_{i=1}^{N}\sum_{t=1}^{T}x_{it}u_{it}\sigma_{i}^{-2}(\frac{1}{T}\sum_{s=1}^{T}(u_{is}^2-Eu_{is}^2)) {\sigma}_{i}^{-2} = o_{P}(1).$$   In fact, straightforward calculations yield
	\begin{equation*} \label{e.3.14}
	EA =   \frac{\sqrt{NT}}{T}\frac{1}{NT}\sum_{i=1}^{N}\sum_{t=1}^{T}E(x_{it}E(u_{it}^3|x_{it}))\sigma_{i}^{-4}.
	\end{equation*}
	Generally, if $u_{it}|x_{it}$ is non-Gaussian and asymmetric, $E(u_{it}^3|x_{it}) \neq 0$.   Hence we require  $N/T \rightarrow 0$ to have $EA \rightarrow 0$.    Hence, to allow for non-Gaussian and asymmetric conditional distributions, in  the GLS setting it turns out $N=o(T)$ is required.
	
	 We shall not explicitly impose  $N=o(T)$ in this paper  as a formal assumption, but instead impose  Assumption \ref{assum4}.
	 On one hand, 	when the distribution of $u_{it}$ is symmetric, we \textit{do not} require $N=o(T)$ because as is shown in the above example, $E(u_{it}^3|x_{it})=0$ is sufficient and holds for symmetric distributions. On the other hand, when $u_{it}$ is non-symmetric, 
	 Assumption \ref{assum4}  then  implicitly requires $N=o(T)$. 
	 Note that $N=o(T)$ is a strong assumption in many microeconomic applications for panel data models. But as illustated in the above simple example, if $u_{it}|x_{it}$ is not symmetric, it is required for feasible GLS even if $\Omega$ is diagonal. One possible approach to weakening this assumption is to remove the higher order bias from $\widehat\Omega$.  Higher order debiasing is a complicated procedure in the presence of general weak dependences. This is left for future research. 
	

	\subsection{Asymptotic results of FGLS}\label{sec3.3}
	We impose the following conditions, regulating the sparsity and serial weak dependence.
	\begin{assum} \label{assum1}
		(i) $\{u_{t},x_{t}\}_{t\geq 1}$ is strictly stationary. In addition, each $u_{t}$ has zero mean vector, and $\{u_{t}\}_{t\geq 1}$ and $\{x_{t}\}_{t\geq 1}$ are independent.\\
		(ii) There are constants $c_{1}, c_{2}> 0$ such that $\lambda_{\min}(\Omega_{h})>c_{1}$ and $\|\Omega_{h}\|_{1} < c_{2}$ for each fixed $h$.\\
		(iii) Exponential tail: There exist $r_{1}, r_{2}>0$ and $b_{1}, b_{2} > 0$, and for any $s > 0, i\leq N$ and $l \leq d$,  
		$$P(|u_{it}| > s) \leq exp(-(s/b_{1})^{r_1}),\quad P(|x_{it,l}|>s) \leq \exp(-(s/b_2)^{r_2}).$$
		(iv) Strong mixing: There exist $\kappa\in(0,1)$ such that $ r_{1}^{-1}+r_{2}^{-1}+\kappa^{-1}>1$, and $C>0$ such that for all $T>0$, 		$$
		\sup\limits_{A\in \mathcal{F}_{-\infty}^0, B \in \mathcal{F}_{T}^{\infty}}|P(A)P(B)-P(AB)| < \exp(-CT^{\kappa}),
		$$
		where $\mathcal{F}_{-\infty}^0$ and $\mathcal{F}_{T}^{\infty}$ denote the $\sigma$-algebras generated by $\{(x_{t},u_{t}) : t \leq 0\}$ and $\{(x_{t},u_{t}) : t \geq T\}$ respectively.
	\end{assum}

	Condition (ii) requires that $\Omega_{h}$ be well conditioned. Condition (iii) ensures the Bernstein-type inequality for weakly dependent data, which requires
	the underlying distributions to be thin-tailed. Condition (iv) is the standard $\alpha$-mixing condition, adapted to the large-$N$ panel. In addition, we impose the following regularity conditions. 
	\begin{assum} \label{assum2}
		(i) There exists a constant $C>0$ such that for all $i\leq N$ and $t\leq T$, $E\|x_{it}\|^{4}<C$ and $Eu_{it}^{4}<C$. \\
		(ii) Define $\xi_{T}(L) = \max_{t \leq T}\sum_{|h| > L} \|Eu_tu_{t-h}'\|$. Then  $\xi_{T}(L)\rightarrow 0.$\\
		(iii) Define $f_{T}(L)=\max_{t \leq T}\sum_{|h|\leq L}\|Eu_{t}u_{t-h}'(1-\omega(|h|,L))\|$. Then $f_{T}(L) \rightarrow 0$.
	\end{assum}

	Assumption \ref{assum2} allows us to prove the convergence rate of the covariance matrix estimator. Condition (ii) is an extension of the standard weak serial dependence condition to the high-dimensional case in panel data literature. It allows us to employ banding or Newey-West trunction procedure. Condition (iii) is well satisfied by various kernel functions for the HAC-type estimator. For the Bartlett kernel, for example,
	$$
	\max_{t \leq T}\sum_{|h|\leq L}\|Eu_{t}u_{t-h}'(1-\omega(|h|,L))\| \leq \frac{1}{L}\max_{t \leq T}\sum_{|h|=0}^{\infty}\|Eu_{t}u_{t-h}'\||h|
	$$
	converges to zero as $L\rightarrow \infty$ as long as $\max_{t \leq T}\sum_{|h|=0}^{\infty}\|Eu_{t}u_{t-h}'\||h| < \infty$.
	
	In this paper, we assume $\Omega_{h}$ to be a sparse matrix for each $h$ and impose similar conditions as those in  \cite{bickel2008a} and \cite{fan2013large}: write $\Omega_{h} = (\Omega_{h,ij})_{N\times N}$, where $\Omega_{h,ij}=Eu_{it}u_{j,t-h}$. For some $q \in [0,1)$, we define
	$$
	m_{N}=\max_{|h| \leq L}\max_{i\leq N}\sum_{j=1}^{N}|\Omega_{h,ij}|^q,
	$$
	as a measurement of the sparsity. We would require that  $	m_{N}$ should be either fixed or grow slowly as $N \rightarrow \infty$. In particular, when $q=0$,
	$m_{N} = \max_{|h|\leq L}\max_{i \leq N}\sum_{j=1}^{N}1(\Omega_{h,ij}\neq 0)$, which corresponds to the exact sparsity case.
	
		Let $$\gamma_{T} = \sqrt{\log(LN)/T}.$$
	
	The following theorem shows the convergence rate of the estimated large covariance matrix. For technical simplicity, we assume that there is no fixed effects so that we do not take the de-meaning procedure. Extending to the more complete estimators with de-meaning is straightforward, but should require more technical arguments to show that the effect from added dependences due to the de-meaning is negligible. 
	\begin{thm}\label{thm.1}
		Under the Assumptions \ref{assum1}-\ref{assum2}, when $\|\Omega^{-1}\|_{1} = O(1)$, for $q\in[0,1)$ such that $Lm_{N}\gamma_{T}^{1-q}=o(1)$,
		\begin{equation*}
		\|\widehat{\Omega}-\Omega\| =O_{P}(Lm_{N}\gamma_{T}^{1-q}+ \xi_{T}(L) +f_{T}(L))=\|\widehat{\Omega}^{-1}-\Omega^{-1}\|.
		\end{equation*}	
	 
	\end{thm}

	The following conditions are required to prove Step 1 in the previous section. 
	
	\begin{assum} \label{assum3}
		For any $NT \times NT$ matrix $M$, we denote $(M)_{ts,ij}$ as the $(i,j)$th element of the $(t,s)$th block of the matrix $M$.\\
		(i) $\sum_{|h|>L}\|\Omega_{h}\|_{1} = O(L^{-\alpha})$,  for a constant $\alpha>0$.\\
		(ii) $\max_{i\leq N,t\leq T}\sum_{s=1}^{T}\sum_{j=1}^{N}|(\Omega^{-1})_{ts,ij}| = O(1)$.\\	
		(iii) There is $q\in [0,1)$ such that   $Lm_{N}\gamma_{T}^{1-q}=o(1)$ holds. In addition,  
		$$
		\sqrt{T}L^{2}m_{N}^2\gamma_{T}^{3-2q} =o(1),\;\;  L^{-\alpha}T\sqrt{NT}m_{N}\gamma_{T}^{1-q} =o(1).
		$$   	
		(iv) Define $\gamma^{*} = Lm_{N}\gamma_{T}^{1-q}+\xi_{T}(L) + f_{T}(L)$. Then $\sqrt{NT}\gamma^{*3}=o(1).$
	\end{assum}

	 Conditions (i)-(ii) require the weak cross-sectional correlations. Conditions (iii)-(iv) are the sparsity assumptions. In addition, the sparsity assumptions assume that $m_{N}$ should not be too large.  

	\begin{remark}
		To understand Assumption \ref{assum3}, consider $q=0$ as a simple case. Then conidtions (iii)-(iv) reduce to, for some positive constant $\alpha$,
		\begin{equation*}
		NT^{2}m_{N}^{2}\log(LN) = o(L^{2\alpha}), \\
		\;\; NL^{6}m_{N}^{6}(\log(LN))^{3} = o(T^2).
		\end{equation*}
	If $m_{N}=O(1)$, then these conditions are  simplified to 
	$$
	NT^2\log(LN)=o(L^{2\alpha}),\quad NL^6(\log (LN))^3=o(T^2).
	$$
	
	\end{remark}

	\begin{pro} \label{proposition1}
		Under the Assumption \ref{assum1}-\ref{assum2}, for $q \in[0,1)$ and $\alpha>0$ such that Assumption \ref{assum3} holds,
		\begin{equation*} \label{proposition}
		\sqrt{NT}(\widehat{\beta}_{FGLS}-\beta) = \Gamma^{-1}\left(\frac{1}{\sqrt{NT}}X'\Omega^{-1}U\right)+ \Gamma^{-1}\left(\frac{1}{\sqrt{NT}}X'\Omega^{-1}(\widehat{\Omega}-\Omega)\Omega^{-1}U\right)+o_{P}(1),
		\end{equation*}
		where $\Gamma = E(X'\Omega^{-1}X/NT)$. 
	\end{pro}
	
	In addition, we impose the following assumption, which allows us to prove that the second term on the right hand side in the above equation  is $o_{P}(1)$.
	\begin{assum} \label{assum4}
		Let $A_{b_h} = \{(i,j) : |Eu_{it}u_{j,t-h}| \neq 0\}$. Then 
		\begin{equation}\label{empirical process}
		\left\|  \frac{1}{\sqrt{NT}}\sum_{h=0}^{L}\sum_{i,j \in A_{b_h}}\mathbb{G}_{T,ij}^{1}(h)\mathbb{G}_{T,ij}^{2}(h) \right\|= o_{P}(1),
		\end{equation}
		 where $\mathbb{G}_{T,ij}^{1}(h) = \frac{1}{\sqrt{T}}\sum_{t=h+1}^{T}(u_{it}u_{j,t-h}-Eu_{it}u_{j,t-h})$ and $\mathbb{G}_{T,ij}^{2}(h) = \frac{1}{\sqrt{T}}\sum_{t=h+1}^{T}w_{it}\varepsilon_{j,t-h}$. 
	\end{assum}
	The right hand side of equation (\ref{e.3.16}) can be written as the equation (\ref{empirical process}).  Then we have the following limiting distribution by using the result of Theorem \ref{thm.1}.
	
	\begin{thm}\label{thm.2}
		Suppose $\var(U|X)=\var(U)=\Omega$. Under the Assumptions \ref{assum1}-\ref{assum4}, for $q \in[0,1)$ and $\alpha>0$ such that Assumption \ref{assum3} holds, as $N, T \rightarrow \infty$, 
		\begin{equation*}
		\sqrt{NT}(\widehat{\beta}_{FGLS}-\beta) \overset{d}{\to} \mathcal{N}(0,\Gamma^{-1}),
		\end{equation*}
		where $\Gamma = E(X'\Omega^{-1}X/NT)$. The consistent estimator of $\Gamma$ is $\widehat{\Gamma} = X'\widehat{\Omega}^{-1}X/NT$. 
	\end{thm}
	The asymptotic variance of the FGLS estimator is $\text{Avar}(\widehat{\beta}_{FGLS})= \Gamma^{-1}/NT$, and an estimator of it is $(X'\widehat{\Omega}^{-1}X)^{-1}$. Asymptotic standard errors can be obtained in the usual fashion from the asymptotic variance estimates.

	\section{Monte Carlo evidence} \label{simulation}
	\subsection{DGP and methods}
	In this section we compare the proposed FGLS estimator with OLS estimator.  We consider the fixed effect linear regression model, although this paper focuses on the simple linear model for technical simplicity. Hence the de-meaning procedure is applied first. The data generating process (DGP) used for the simulations is given by
	\begin{equation*}
	y_{it} = \alpha_{i} + \mu_{t}+\beta_{0}x_{it} + u_{it},
	\end{equation*}
	where the true $\beta_{0} =1$ and fixed effects $\alpha_{i}, \mu_{t}$ are generated from $\mathcal{N}(0,0.5)$.  The DGP allows for serial and cross-sectional correlation in both $x_{it}$ and $u_{it}$, which are generated by $(NT) \times (NT)$ covariance matrices, $\Omega_{X}$ and $\Omega_{U}$, as follows: let $R_{\eta} = (R_{\eta,ij})$ denote an $N \times N$ block diagonal correlation matrix. We fix the number of clusters as $G=25$. Hence, each diagonal block is a $N/G \times N/G$ matrix with the off-diagonal entries $(i,j)$ in the same cluster, $R_{\eta,ij}$ for $i \neq j$, which are generated from  i.i.d. Uniform$(0,\gamma)$. In this study, we set the level of cross-sectional correlation in each cluster as $\gamma=0.3,$ or $0.7$.  For the cross-sectional heteroskedasticity, let $D=\text{diag}\{d_{i}\}$, where $\{d_{i}\}_{i \leq N}$ are i.i.d. Uniform(1,$m$). Finally, we define the $N \times N$ covariance matrix of $u_{t}$ as $\Sigma_{u}=DR_{\eta}D$. In this case, we report results when $m=\sqrt{5}$. For the covariance matrix of the regressor, we simply set $\Sigma_{x} = R_{\eta}$, which does not have heteroskedasticity. 
	
	Now we introduce $i$-dependent serial correlation for the regressor and the error as follows: first let $\sigma_{ii} = \rho_{i}$ if $i=j$ and $\sigma_{ij} = \rho_{i}\rho_{j}$ if $i\neq j$. Then we define the $(NT) \times (NT)$ covariance matrix, $\Omega_{U} = (\Omega_{t,s})$. The $(t,s)$th block is an $N \times N$ covariance matrix, given by $\Omega_{t,s}= (\Omega_{t,s}(i,j))$, where $\Omega_{t,s}(i,j) = \Sigma_{u,ij}\sigma_{ij}^{|t-s|}$. The large covariance matrix of the regressor, $\Omega_{X}$, is generated similarly. The level of $i$-dependent $\rho_{i}$ of the regressor and the error is generated from i.i.d. Uniform$(0,0.6)$, seperately.
	 
	 Note that the $(t,s)$th block covariance decays exponentially as $|t-s|$ increases. Finally we generate the $NT \times 1$ vectors $(u_{1}',...,u_{T}')' = \Omega_{U}^{1/2}\zeta$, where $\zeta$ is an $NT \times 1$ vector, whose entries are generated from i.i.d. $\mathcal{N}(0,5)$. Similarly, the regressor is generated by  $(x_{1}',...,x_{T}')' = \Omega_{X}^{1/2}\xi$, where $\xi$ is an $NT \times 1$ vector, whose entries are generated from i.i.d. $\mathcal{N}(0,1)$. Note that $x_{it}$ is uncorrelated with $u_{it}$.

	 In this numerical study, we use sample sizes $N=50, 100$ and $T=50, 100, 150$, and the simulation is replicated for one thousand times in all cases.\footnote{The procedure of proposed estimators require use of an $NT \times NT$ matrix as discussed in Section \ref{implemenation}. Indeed, when $NT$ is large, the procedure appears to be computationally demanding. Hence, we focus on the small sample size in this study.} For each $\{N, T\}$ combination, we set the bandwidth $L = 3$ in all cases. The threshold constant, $M$, is obtained by the cross-validation method as suggested in Section \ref{tunning}. For instance, when $T=100$, the number of folds to split is $\log(100) \approx 5$. In general, the cross-validation chooses $M$ between 1.4 and 1.8. Interestingly, as the level of cross-section correlation increases, the cross-validation tends to choose smaller M, so that the number of non-thresholded elements increases. Hence it takes into account the strength of cross-sectional correlation. We use the Bartlett kernel for our FGLS estimator. Results are summarized in Tables \ref{weak}-\ref{strong}.

	\subsection{Results}
	Tables \ref{weak}-\ref{strong} present the simulation results, where each table corresponds to a different level of cross-sectional correlation, $\gamma = \{0.3, 0.7\}$. In each table, the mean and standard deviation of the estimators are reported. FGLS(Diag) refers to the FGLS estimator using the diagonal covariance matrix, which only takes into account heteroskedasticity. RMSE is the ratio of the mean squared error of FGLS to that of OLS. The mean and standard deviation of the estimated standard errors for OLS and FGLS are also reported. The robust unknown clustered standard error, suggested by \cite{bai2019olsse}, is used for OLS. For FGLS, we report the results of the standard error as introduced in Theorem \ref{thm.2}. The difference between the standard deviation of the estimators and the mean of standard errors can be explained as the bias of estimated standard errors. In addition, we present null rejection probabilities for the 5\% level tests using the traditional $\mathcal{N}(0,1)$ critical value based on each standard errors. 
		
	According to Tables \ref{weak}-\ref{strong}, we see that both methods are almost unbiased, while our proposed FGLS has indeed smaller standard deviation of $\widehat{\beta}$ than that of OLS and FGLS(Diag). In all cases, the RMSE of our proposed FGLS is significantly smaller than one. Hence the results confirm that the FGLS estimator is more efficient than the OLS and the FGLS(Diag) estimators in presence of heteroskedasticity, serial and cross-sectional correlations. Regarding the $t$-test, in Table \ref{weak}, the rejection probabilities of FGLS and OLS are close to 0.05 when $T$ is large, while those of FGLS(Diag) tend to over-reject. Since the FGLS(Diag) estimator does not take into account the serial and the cross-sectional correlations, its standard errors are underestimated. On the other hand, in Table \ref{strong}, we find that the standard errors of all estimators are underestimated and the t-test rejection probabilities are much larger than 0.05, especially when $T$ is relatively smaller than $N$ (e.g., $N=100$ and $T=50$). This is due to the strong cross-sectional correlation within clusters. However, the rejection probabilities of FGLS and OLS are much smaller than those of FGLS(Diag). 
	In summary, FGLS does improve efficiency in terms of mean squared error; also we obtain unbiased standard error estimator and appropriate rejection rate as $T$ increases.

	\section{Empirical study: Effects of divorce law reforms on divorce rates} \label{application}

	In the literature, the cause of the sharp increase in the U.S. divorce rate in the 1960-1970s is an important research question. During 1970s, more than half of states in the U.S. liberalized the divorce system, and the effects of reforms on divorce rates have been investigated by many  such as \cite{allen1992marriage} and \cite{peters1986marriage}. With controls for state and year fixed effects, \cite{friedberg1998} suggested that state law reforms siginificantly increased divorce rates. Also, she assumed that unilateral divorce laws affected divorce rates permanently. However, divorce rates from 1975 have been subsequently decreasing according to empirical evidence. Therefore the question of whether law reforms also affect the divorce rate decrease has arisen. \cite{wolfers2006did} revisited this question by using a treatment effect panel data model, and identified only temporal effects of reforms on divorce rates. In particular, he used dummy variables for the first two years after the reforms, 3-4 years, 5-6 years, and so on. More specifically, the following fixed effect panel data model was considered:
	\begin{equation}
	y_{it} =  \alpha_{i} + \mu_{t} +  \sum_{k=1}^{8}\beta_{k}X_{it,k} + \delta_{i}t + u_{it},
	\end{equation}
	where $y_{it}$ is the divorce rate for state $i$ and year $t$, $\alpha_{i}$ a state fixed effect, $\mu_{t}$ a time fixed effects, and $\delta_{i}t$ a linear time trend with unknown coefficient $\delta_{i}$. $X_{it}$ is a binary regressor which denotes the treatment effect $2k$ years after the reform. \cite{wolfers2006did} suggested that ``the divorce rate rose sharply following the adoption of unilateral divorce laws, but this rise was reversed within about a decade". He also concluded that ``15 years after reform the divorce rate is lower as a result of the adoption of unilateral divorce, although it is hard to draw any strong conclusions about long-run effects’’.

	Both \cite{friedberg1998} and \cite{wolfers2006did} used a weighted model by muliplying all variables by the square root of  state population. In addition, they used ordinary OLS standard error, which does not take into account heteroskedasticity, serial and cross-sectional correlations. However, standard errors might be biased when one disregards these correlations. Therefore, we re-estimated the model of \cite{wolfers2006did} using the proposed FGLS method and OLS with the heteroskedastic standard errors of \cite{white1980heteroskedasticity}, the clustered standard error of \cite{arellano1987}, and the robust standard error of \cite{bai2019olsse}.
	
	The same dataset as in \cite{wolfers2006did} is used, which includes the divorce rate, state-level reform years, binary regressors, and state population. Due to missing observations around divorce law reforms, we exclude Indiana, New Mexico and Louisiana. As a result, we obtain balanced panel data from 1956 to 1988 for 48 states. We fit the models both with and without linear time trend, and use OLS and FGLS in each model to estimate $\beta$. In the FGLS estimation, we set bandwidth $L=3$  as proposed by \cite{newey1994} ($L=4(T/100)^{2/9}$). The thresholding values are chosen by the cross-validation method as discussed in Section \ref{tunning}, more specifically, $M=1.8$ and $M=1.9$ for the model with and without linear time trends, respectively.  The Bartlett kernel is used in the OLS robust standard error and FGLS estimation. The estimated $\beta_{1}, \cdots, \beta_{8}$ with and without linear time trend and standard errors are summarized in Table \ref{div} below.
	
	The OLS and FGLS estimates in both models are similar to each other. The results show that divorce rates rose soon after the law reform. However, within a decade, divorce rates had fallen over time. Interestingly, FGLS confirms the negative effects of the law reforms on the divorce rates, specifically, 11-15+ years after the reform in the model with state-specific linear time trends, and 9-15+ years after the reform in the model without state-specific linear time trends. In addition, the FGLS estimates for 1-6 and 1-4 years are positive and statistically significant in the models with and without linear time trends, respectively. For OLS, the coefficient estimates for 3-4 and 7-15+ are significant in the model without linear time trends based on $se_{BCL}$. In contrast, the OLS estimates are  statistically significant only for 1-4 years when a linear time trend is added. According to the clustered standard error, $se_{CX}$, note that only 11-15+ are statistically significant in the model without trends.

	According to OLS and FGLS estimation results with and without a linear time trend, we make the following conlcusion: in the first 8 years, the overall trend of divorce rate is increasing, but the law reform reduces the divorce rate after 3-4 years. However, 8 years after the reform, we observe that the law reform has a negative effect on divorce rate. Note that  \cite{wolfers2006did} de-emphasized the negative coefficient at the end of the periods, as these are not robust to inclusion of state-specific quadratic trends, which we did not employ in this paper. Overall, the results of FGLS estimates are consistent with  \cite{wolfers2006did}.

	\section{Conclusions} \label{conclusion}
	In this paper, we propose a large covariance matrix estimator and a modified version of FGLS that  takes into account both serial and cross-sectional correlations in linear panel models that are robust to heteroskedasticity, serial and cross-sectional correlations. The covariance matrix estimator is asymptotically unbiased with an improved convergence rate. It is shown to be more efficient than other existing methods in panel data literature. From simulated experiments, we confirmed that our FGLS estimates are more efficient than OLS estimates.

	\begin{table}[h!tbp] 
		\centering
		\caption{Performance of estimated $\beta_{0}$; true $\beta_{0}=1$; $i$-dependent serial correlation and weak cross-sectional correlation ($\gamma=0.3$). } \label{weak}
		\begin{tabular}{llrrrrrrrrrrr}
			\toprule\toprule
			&&       OLS & \multicolumn{2}{c}{FGLS}   & & OLS & \multicolumn{2}{c}{FGLS} & &  OLS & \multicolumn{2}{c}{FGLS} \\
			\cline{4-6} 	\cline{8-9}		\cline{12-13}
			N  &   T     &  &Diag&Our &	& &Diag&Our&	& &Diag&Our \\	
			\midrule
			&& \multicolumn{3}{c}{mean($\widehat{\beta}$)} && \multicolumn{3}{c}{std($\widehat{\beta}$)} && \multicolumn{3}{c}{RMSE} \\ 
			\cline{3-5} \cline{7-9} \cline{11-13}
			\noalign{\vskip 2mm} 			
			50	&	50	&	1.001	&	1.002	&	1.001	&&	0.080	&	0.075	&	0.069	&&	1.000	&	0.883	&	0.740	\\
			&	100	&	0.999	&	0.999	&	0.999	&&	0.061	&	0.055	&	0.050	&&	1.000	&	0.823	&	0.680	\\
			&	150	&	1.000	&	1.000	&	1.000	&&	0.045	&	0.041	&	0.038	&&	1.000	&	0.842	&	0.710	\\
		100	&	50	&	1.000	&	1.000	&	1.001	&&	0.058	&	0.053	&	0.050	&&	1.000	&	0.842	&	0.745	\\
			&	100	&	0.999	&	0.998	&	0.998	&&	0.041	&	0.037	&	0.034	&&	1.000	&	0.793	&	0.690	\\
			&	150	&	1.000	&	1.000	&	1.000	&&	0.034	&	0.029	&	0.027	&&	1.000	&	0.749	&	0.628	\\					
			\noalign{\vskip 2mm} 
			&&\multicolumn{3}{c}{mean(s.e.)}&& \multicolumn{3}{c}{std(s.e.)} & & \multicolumn{3}{c}{t-test rejection prob.}\\ 
			\cline{3-5} \cline{7-9} \cline{11-13}
			\noalign{\vskip 2mm} 			
			
			50	&	50	&	0.079	&	0.066	&	0.065	&&	0.004	&	0.002	&	0.002	&&	0.054	&	0.084	&	0.068	\\
				&	100	&	0.058	&	0.048	&	0.047	&&	0.002	&	0.001	&	0.001	&&	0.052	&	0.090	&	0.073	\\
				&	150	&	0.047	&	0.039	&	0.039	&&	0.001	&	0.001	&	0.001	&&	0.047	&	0.068	&	0.043	\\
			100	&	50	&	0.058	&	0.048	&	0.046	&&	0.002	&	0.001	&	0.001	&&	0.058	&	0.083	&	0.069	\\
				&	100	&	0.040	&	0.034	&	0.033	&&	0.001	&	0.000	&	0.000	&&	0.053	&	0.069	&	0.067	\\
				&	150	&	0.032	&	0.027	&	0.026	&&	0.001	&	0.000	&	0.000	&&	0.069	&	0.077	&	0.059	\\
			\bottomrule\bottomrule
		\end{tabular}\\
		\begin{tablenotes}
			\item \textbf{Note:} OLS and FGLS comparison. RMSE is the ratio of the mean squared error of FGLS to that of OLS. The $t$-test rejection prob. is $t$-test rejection rates for 5\% level tests. Robust standard error suggested by \cite{bai2019olsse} is used for OLS. Reported results are based on 1000 replications. The threshold value, $M$, is chosen through the cross-validation method as discussed in Section \ref{tunning}. For the bandwidth, we set $L=3$.
		\end{tablenotes}
		\end{table}

		\begin{table}[h!tbp] 
		\centering
		\caption{Performance of estimated $\beta_{0}$; true $\beta_{0}=1$; $i$-dependent serial correlation and strong cross-sectional correlation ($\gamma=0.7$). } \label{strong}
		\begin{tabular}{llrrrrrrrrrrr}
			\toprule\toprule
			&&       OLS & \multicolumn{2}{c}{FGLS}   & & OLS & \multicolumn{2}{c}{FGLS} & &  OLS & \multicolumn{2}{c}{FGLS} \\
			\cline{4-6} 	\cline{8-9}		\cline{12-13}
			N  &   T     &  &Diag&Our &	& &Diag&Our&	& &Diag&Our \\	
			\midrule
			&& \multicolumn{3}{c}{mean($\widehat{\beta}$)} && \multicolumn{3}{c}{std($\widehat{\beta}$)} && \multicolumn{3}{c}{RMSE} \\ 
			\cline{3-5} \cline{7-9} \cline{11-13}
			\noalign{\vskip 2mm} 			
			50	&	50	&	1.000	&	1.001	&	1.000	&&	0.084	&	0.079	&	0.072	&&	1.000	&	0.883	&	0.744	\\
				&	100	&	0.999	&	1.000	&	1.000	&&	0.063	&	0.057	&	0.052	&&	1.000	&	0.817	&	0.677	\\
				&	150	&	1.002	&	1.003	&	1.003	&&	0.051	&	0.047	&	0.042	&&	1.000	&	0.856	&	0.685	\\
			100	&	50	&	1.001	&	1.000	&	1.001	&&	0.070	&	0.063	&	0.059	&&	1.000	&	0.810	&	0.711	\\
				&	100	&	0.998	&	0.999	&	0.999	&&	0.048	&	0.044	&	0.042	&&	1.000	&	0.840	&	0.742	\\
				&	150	&	1.000	&	1.000	&	1.000	&&	0.038	&	0.033	&	0.030	&&	1.000	&	0.777	&	0.617	\\
			
			\noalign{\vskip 2mm} 
			&&\multicolumn{3}{c}{mean(s.e.)}&& \multicolumn{3}{c}{std(s.e.)} & & \multicolumn{3}{c}{t-test rejection prob.}\\ 
			\cline{3-5} \cline{7-9} \cline{11-13}
			\noalign{\vskip 2mm} 			
			
			50	&	50	&	0.079	&	0.066	&	0.065	&&	0.004	&	0.002	&	0.002	&&	0.065	&	0.100	&	0.082	\\
			&	100	&	0.059	&	0.048	&	0.048	&&	0.002	&	0.001	&	0.001	&&	0.068	&	0.104	&	0.065	\\
			&	150	&	0.047	&	0.039	&	0.039	&&	0.001	&	0.001	&	0.001	&&	0.065	&	0.110	&	0.064	\\
			100	&	50	&	0.059	&	0.048	&	0.048	&&	0.002	&	0.001	&	0.001	&&	0.105	&	0.143	&	0.125	\\
			&	100	&	0.040	&	0.034	&	0.034	&&	0.001	&	0.000	&	0.000	&&	0.097	&	0.139	&	0.094	\\
			&	150	&	0.032	&	0.027	&	0.028	&&	0.001	&	0.000	&	0.000	&&	0.101	&	0.123	&	0.079	\\	
			\bottomrule\bottomrule
		\end{tabular}\\
		\begin{tablenotes}
			\item \textbf{Note:} See notes to Table \ref{weak}.
		\end{tablenotes}
	\end{table}

	\begin{table}[h!tbp]
		\centering	
		\caption{Empirical application: effects of divorce law refrom with state and year fixed effects: US state level data annual from 1956 to 1988, dependent variable is divorce rate per 1000 persons per year. OLS and FGLS estimates and standard errors (using state population weights).}  \label{div}
		\begin{tabular}{p{2.3cm}p{1.2cm}p{1.2cm}p{1.2cm}p{1.2cm}p{1.2cm}p{1.2cm}}
			\toprule\toprule
			\noalign{\vskip 2mm} 
			Effects: & $\widehat{\beta}_{OLS}$  & $se_{W}$ & $se_{CX}$ & $se_{BCL}$ & $\widehat{\beta}_{FGLS}$ &$ se_{FGLS}$ \\
			\midrule
			\noalign{\vskip 2mm} 
			\multicolumn{7}{c}{Panel A: Without state-specific linear time trends} \\
			\noalign{\vskip 2mm} 
			1\textendash 2 years	&	0.256	&	0.140	&	0.189	&	0.148	&	0.133	&	0.046$^{*}$		\\
			3\textendash 4 years	&	0.209	&	0.081$^{*}$	&	0.159	&	0.089$^{*}$	&	0.165	&	0.056$^{*}$		\\
			5\textendash 6 years	&	0.126	&	0.073	&	0.168	&	0.069	&	0.100	&	0.059	\\
			7\textendash 8 years	&	0.105	&	0.070	&	0.165	&	0.040$^{*}$	&	0.026	&	0.061		\\
			9\textendash 10 years	&	-0.122	&	0.060$^{*}$	&	0.161	&	0.054$^{*}$	&	-0.129	&	0.061$^{*}$	\\
			11\textendash 12 years	&	-0.344	&	0.071$^{*}$	&	0.173$^{*}$	&	0.075$^{*}$	&	-0.253	&	0.062$^{*}$	\\
			13\textendash 14 years	&	-0.496	&	0.074$^{*}$	&	0.188$^{*}$	&	0.062$^{*}$	&	-0.324	&	0.063$^{*}$\\
			15+ years	&	-0.508	&	0.089$^{*}$	&	0.223$^{*}$	&	0.077$^{*}$	&	-0.325	&	0.067$^{*}$	\\
			\noalign{\vskip 2mm} 
			\multicolumn{7}{c}{Panel B: With state-specific linear time trends} \\
			\noalign{\vskip 2mm} 
			1\textendash 2 years	&	0.286	&	0.152	&	0.206	&	0.140$^{*}$	&	0.171	&	0.044$^{*}$	\\
			3\textendash 4 years	&	0.254	&	0.099$^{*}$	&	0.171	&	0.126$^{*}$	&	0.220	&	0.058$^{*}$	\\
			5\textendash 6 years	&	0.186	&	0.102	&	0.206	&	0.143	&	0.175	&	0.067$^{*}$	\\
			7\textendash 8 years	&	0.177	&	0.109	&	0.230	&	0.146	&	0.097	&	0.075		\\
			9\textendash 10 years	&	-0.037	&	0.111	&	0.241	&	0.154	&	-0.073	&	0.082	\\
			11\textendash 12 years	&	-0.247	&	0.128	&	0.268	&	0.183	&	-0.240	&	0.089$^{*}$		\\
			13\textendash 14 years	&	-0.386	&	0.137$^{*}$	&	0.295	&	0.209	&	-0.329	&	0.098$^{*}$\\
			15+ years	&	-0.414	&	0.158$^{*}$	&	0.337	&	0.243	&	-0.382	&	0.108$^{*}$	\\
			\bottomrule\bottomrule
		\end{tabular}\\
		\begin{tablenotes}
			\item \textbf{Note:} Standard errors with asterisks indicate significance at 5\% level using $N(0,1)$ critical values. For OLS standard errors, $se_{W}$ and $se_{CX}$ refer to the heteroskedastic standard errors by \cite{white1980heteroskedasticity} and the clustered standard errors by \cite{arellano1987}, respectively; $se_{BCL}$ is the robust standard error suggested by \cite{bai2019olsse}. The threshold values for FGLS by the cross-validation are $M=1.9$ and $M=1.8$ for Panel A and B, respectively.
		\end{tablenotes}
	\end{table}

	\clearpage
	\newpage
	
	\appendix
	\section{Appendix} \label{appendix}
	
	Throughout the proof, $\max_{i}$, $\max_{t}$, $\max_{h}$, $\max_{ij}$, and $\max_{it}$ denote $\max_{i\leq N}$, $\max_{t\leq T}$, $\max_{h\leq L}$, $\max_{i\leq N,j\leq N}$, and $\max_{i\leq N, t\leq T}$ respectively. In addition, for technical simplicity, we assume that there is no fixed effects so that we do not take the de-meaning procedure. Extending to the more complete estimators with de-meaning is straightforward, but should require more technical arguments to show that the effect from added dependences due to the de-meaning is negligible. 
	
	\subsection{Proofs of Theorem \ref{thm.1}-\ref{thm.2}.}
	\begin{lem} \label{lem1}
		Under the Assumptions \ref{assum1}-\ref{assum2}, for $\gamma_{T}=\sqrt{\frac{\log(LN)}{T}}$, for $h\geq 0$,
		\begin{equation*}
		\max_{h\leq L}\max_{i,j \leq N} \left\|\frac{1}{T}\sum_{t=h+1}^{T}x_{it}u_{j,t-h}\right\|= O_{P}(\gamma_{T}).
		\end{equation*}
	\end{lem}
	\begin{proof}
	Let $\gamma_{ij,t,h}=x_{it}u_{j,t-h}I_{\{h+1 \leq t\leq T\}}$, where $I_A$ is the indicator function of the set $A$. To simplify notation, we assume $d=\dim(x_{it})=1$. 
	By Lemma A.2 of \cite{fan2011high} and Assumption \ref{assum1} (iii), $\gamma_{ij,t,h}$ satisfies
	the exponential tail condition.  We set $\alpha_{T} = \sqrt{\frac{\log(LN)}{T}}$ and $c_{1}^2=3c_{2}$ for $c_{1}, c_{2}>0$. Using Bernstein inequality for weakly dependent data in \cite{merlevede2011bernstein} and the Bonferroni method, we have 
	\begin{align*}
	P\left(\max\limits_{h}\max\limits_{ij}\left|\frac{1}{T}\sum_{t=1}^{T}\gamma_{ij,t,h}\right| > c_{1}\alpha_{T}\right) &\leq LN^2\max\limits_{h \leq L}\max\limits_{ij}P\left(\left|\cfrac{1}{T}\sum\limits_{t=1}^{T}\gamma_{ij,t,h}\right| > c_{1}\alpha_{T}\right) \rightarrow 0. \;\;
	\end{align*}
	Then $\max_{h}\max_{ij} \|\frac{1}{T}\sum_{t=h+1}^{T}x_{it}u_{j,t-h}\| =O_{P}(\sqrt{\frac{\log(LN^2)\max_{ij,h}\frac{1}{T}\sum_{t=1}^T\var(\gamma_{ij,t,h})}{T}})=O_{P}(\sqrt{\frac{\log(LN)}{T}}).$ 
	\end{proof}

	\begin{lem} \label{lem2}
		Under the Assumptions \ref{assum1}-\ref{assum2}, for $\gamma_{T}=\sqrt{\frac{\log(LN)}{T}}$,\\
		(i) $\max_{h\leq L}\max_{i,j \leq N} |\widetilde{R}_{h,ij}-\Omega_{h,ij}| = O_{P}(\gamma_{T} ),$
		where $\widetilde{R}_{h,ij} = \frac{1}{T}\sum_{t=h+1}^{T}\widehat{u}_{it}\widehat{u}_{j,t-h}$ for $h \geq 0$.\\
		(ii) $ \max_{|h|\leq L}\|\widetilde{\Omega}_{h}-\Omega_{h}\|_{1}=O_{P}(m_{N}\gamma_{T}^{1-q})$.
		
	\end{lem}
	\begin{proof}
		(i) First, we  write
		\begin{align*}
		\max_{h\leq L}\max_{i,j \leq N} \left|\frac{1}{T}\sum_{t=h+1}^{T}\widehat{u}_{it}\widehat{u}_{j,t-h}-Eu_{it}u_{j,t-h}\right|  &\leq  \max\limits_{h}\max\limits_{ij} \left|\frac{1}{T}\sum_{t=h+1}^{T}(\widehat{u}_{it}\widehat{u}_{j,t-h}-u_{it}u_{j,t-h})\right|\\
		& \;\;+\max\limits_{h}\max\limits_{ij} \left|\frac{1}{T}\sum_{t=h+1}^{T}u_{it}u_{j,t-h}-Eu_{it}u_{j,t-h}\right|\\
		& \;\;+\max\limits_{h}\max\limits_{ij}\frac{L}{T}|Eu_{it}u_{j,t-h}|\\
		& \equiv a_1+a_2+a_3.
		\end{align*}
		Then,  $a_1$ is bounded by $a_{11} + a_{12} + a_{13}$, where
		\begin{align*}
		a_{11} & \equiv \max\limits_{h}\max\limits_{ij} \left|\frac{1}{T}\sum_{t=h+1}^{T}(\widehat{u}_{it}-u_{it})(\widehat{u}_{j,t-h}-u_{j,t-h})\right|\\
		a_{12} & \equiv \max\limits_{h}\max\limits_{ij} \left|\frac{1}{T}\sum_{t=h+1}^{T}(\widehat{u}_{it}-u_{it})u_{j,t-h}\right|\\
		a_{13} & \equiv \max\limits_{h}\max\limits_{ij} \left|\frac{1}{T}\sum_{t=h+1}^{T}u_{it}(\widehat{u}_{j,t-h}-u_{j,t-h})\right|.
		\end{align*}
		First, by the Cauchy-Schwarz inequality,
		\begin{align*}
		a_{11} & = \max_{h}\max_{ij}\left|\frac{1}{T}\sum_{t=h+1}^{T}x_{it}'(\widehat{\beta}-\beta)x_{j,t-h}'(\widehat{\beta}-\beta)\right|\\
		&\leq\|\widehat{\beta}-\beta\|^{2}\max_{h}\max_{ij}\frac{1}{T}\sum_{t=h+1}^{T}\|x_{it}\|\|x_{j,t-h}\|\\
		&\leq  O_{P}\left(\cfrac{1}{NT}\right) \max\limits_{i}\cfrac{1}{T}\sum\limits_{t=1}^{T}\|x_{it}\|^2 =  O_{P}\left(\cfrac{1}{NT}\right).
		\end{align*}
		Note that $\max_{i}\frac{1}{T}\sum_{t=1}^{T}\|x_{it}\|^2 $ is bounded by the exponential tail condition and Bernstein's inequality using the same argument of Lemma 3.1 of \cite{fan2011high}.
		Next, by Lemma \ref{lem1},
		\begin{align*}
		a_{12} & \leq  \|\beta-\widehat{\beta}\|\max\limits_{h}\max\limits_{ij}  \left\|\cfrac{1}{T}\sum\limits_{t=h+1}^{T}x_{it}u_{j,t-h}\right\|\\
		& \leq O_{P}(\frac{1}{\sqrt{NT}})\max\limits_{h}\max\limits_{ij}  \left\|\cfrac{1}{T}\sum\limits_{t=h+1}^{T}x_{it}u_{j,t-h}\right\| = O_{P}\left(\frac{1}{T}\sqrt{\frac{\log(LN)}{N}}\right).
		\end{align*}
		Similarly, $a_{13}$ is bounded using the same argument. Then, we have $a_{1} = O_{P}(\frac{1}{T}\sqrt{\frac{\log(LN)}{N}})$.\\
		
		Next, we let $Z_{h,ij,t} = u_{it}u_{j,t-h}-Eu_{it}u_{j,t-h}$, which satisfies the exponential tail condition by Assumption \ref{assum1} and Lemma A.2 of \cite{fan2011high}. Then 
		$a_{2}$ can be written as $\max_{h}\max_{ij}|\frac{1}{T}\sum_{t}Z_{h,ij,t}|$.  Set $\alpha_{T} = \sqrt{\frac{\log(LN)}{T}}$ and $c_{1}^2=3c_{2}$ for $c_{1}, c_{2}>0$. Then, using Bernstein's inequality in \cite{merlevede2011bernstein} and the same argument as in the proof of Lemma \ref{lem1}, 
		\begin{align*}
		P\left(\max\limits_{h \leq L}\max\limits_{ij}\left|\frac{1}{T}\sum_{t=1}^{T}Z_{h,ij,t}\right| > c_{1}\alpha_{T}\right) &\leq LN^2\max\limits_{h \leq L}\max\limits_{ij}P\left(\left|\cfrac{1}{T}\sum\limits_{t=1}^{T}Z_{h,ij,t}\right| > c_{1}\alpha_{T}\right) \rightarrow 0. \;\;
		\end{align*}
		Hence, we have $a_2  = O_{P}(\sqrt{\frac{\log(LN)}{T}})$.	In addition,  $a_{3} = O_{P}(\frac{L}{T})$, which can be proved easily. Together,  
		\begin{align*}
		\max_{h}\max_{ij} \left|\frac{1}{T}\sum_{t=h+1}^{T}\widehat{u}_{it}\widehat{u}_{j,t-h}-Eu_{it}u_{j,t-h}\right| = O_{P}\left(\sqrt{\frac{\log(LN)}{T}}\right).
		\end{align*}

	(ii) Following Theorem 5 of \cite{fan2013large}, we then have $ \max_{|h|\leq L}\|\widetilde{\Omega}_{h}-\Omega_{h}\|_{1}=O_{P}(m_{N} \gamma_{T}^{1-q})$, where $\widetilde{\Omega}_{h}$ is defined in (\ref{e.3.8}).
	\end{proof}

	\begin{lem} \label{lem3}
		For $h\leq L$ and $v \leq L$, let $Q_{imp}^{hv} = \sum_{t=1}^{T}w_{it}\sum_{j=1}^{N}(\Omega^{-1})_{t+h,m+v,jp}$.
		Then, under the Assumption \ref{assum1}-\ref{assum2},
		$$
		\max_{h,v\leq L}\max_{i,p\leq N}\left\|\frac{1}{NT}\sum_{q=1}^{N}\sum_{m=1}^{T}\varepsilon_{qm}Q_{imp}^{hv}\right\| = O_{P}\left(\sqrt{\frac{\log(LN)}{NT}}\right).$$
	\end{lem}
	\begin{proof}
		First, we define $W=\Omega^{-1}X=(w_{1}',\cdots, w_{T}')'$ ($NT\times d$), and $\varepsilon = \Omega^{-1}U=(\varepsilon_{1}',\cdots, \varepsilon_{T}')'$ ($NT\times 1$).  Let $w_{it}'$ and $\varepsilon_{it}$ denote the $i$th row of $w_{t}$ and the $i$th element of $\varepsilon_{t}$, respectively. For simplicity, we assume $d=1$. Let $\zeta_{ipq,mhv} = \varepsilon_{qm}Q_{imp}^{hv}$. Note that due to $\|\Omega^{-1}\|_{1} < \infty$, we know $E(Q_{imp}^{hv})^{2} < \infty$. Then $\max_{hv}\max_{ip}\frac{1}{NT}\sum_{q=1}^{T}\sum_{m=1}^{T}\var(\zeta_{ipq,mhv})$ is bounded. Set $\alpha_{NT} = \sqrt{\frac{\log(LN)}{NT}}$ and $c_{1}^{2}=3c_{2}$ for $c_{1}, c_{2}> 0$. Then, using Bernstein's inequality and the same argument as in the proof of Lemma \ref{lem1},  
		\begin{align*}
		P\left(\max\limits_{h,v \leq L}\max\limits_{ip}\left|\frac{1}{NT}\sum_{q=1}^{N}\sum_{m=1}^{T}\zeta_{ipq,mhv}\right| > c_{1}\alpha_{NT}\right)
		&\leq L^2N^2\max\limits_{h,v\leq L}\max\limits_{ij}P\left(\left|\frac{1}{NT}\sum_{q=1}^{N}\sum_{m=1}^{T}\zeta_{ipq,mhv}\right| > c_{1}\alpha_{NT}\right)\\
		&\rightarrow 0. \;\;
		\end{align*}
		Therefore, we have $\max_{h,v\leq L}\max_{i,p\leq N}\|\frac{1}{NT}\sum_{q=1}^{N}\sum_{m=1}^{T}\varepsilon_{qm}Q_{imp}^{hv}\|  =O_{P}(\sqrt{\frac{\log(LN)}{NT}}).$ 
	\end{proof}

\begin{lem} \label{lem4}
	Consider a symmetric block matrix $A = (A_{ij})\in \mathbb{R}^{dn\times dn}$ where $A_{ij}\in \mathbb{R}^{d\times d}$. Then $$\|A\| \leq  \max_{i}\sum_{j=1}^{n}\|A_{ij}\|.$$
\end{lem}
\begin{proof}
	Suppose $\sigma(\cdot)$ is the spectrum of a matrix, which is the set of its eigenvalues. By Gershgorin's Theorem for block matrices (see \cite{salas1999gershgorin}), if we define
	\begin{equation*}
	G_{i} \equiv \sigma(A_{ii})\cup T_{i},
	\end{equation*}
	where $T_{i} = \{\lambda \notin \sigma(A_{ii}): \|(A_{ii}-\lambda I_{d})^{-1}\|^{-1} \leq \sum_{j=1, j \neq i}^{n} \|A_{ij}\|\}$, then
	\begin{equation*}
	\sigma(A) \subset \bigcup_{i=1}^{n}G_{i}. 
	\end{equation*}
	Note that this theorem means the eigenvalue of $A$ either equals $\sigma(A_{ii})$ or in that specific region.
	
	Let $\lambda \in \cup_{i=1}^{n}G_{i}$. If $\lambda \in \sigma(A_{ii})$ for some $i$, then $|\lambda| \leq \|A_{ii}\| \leq \max_{i}\sum_{j=1}^{n}\|A_{ij}\|$. If $\lambda \notin \sigma(A_{ii})$ for all $i$, then we know  $\lambda \in T_{i}$ for some $i$. Now we consider two cases: (i) $\|A_{ii}\| < |\lambda|$, and (ii) $\|A_{ii}\| \geq |\lambda|$, where $i$ such that $\lambda \in T_{i}$.
	For the case of (i), note that if a matrix $M$ is such that  $\|M\| <1$, then
	\begin{equation} \label{inequ}
	\frac{1}{1+\|M\|} \leq \|(I-M)^{-1}\| \leq \frac{1}{1-\|M\|}.
	\end{equation}
	Then we have 
	\begin{align*}
	|\lambda| -\|A_{ii}\| &\leq |\lambda|\left(1-\frac{\|A_{ii}\|}{|\lambda|}\right) \\
	& \leq  |\lambda| \left\|\left(I_{d}-\frac{A_{ii}}{|\lambda|}\right)^{-1} \right\|^{-1} =\|(|\lambda| I_{d}-A_{ii})^{-1}\|^{-1}\\
	&\leq \sum_{j=1, j\neq i}^{n}\|A_{ij}\| . 
	\end{align*}
	Therefore, we have $|\lambda| \leq \sum_{j=1}^{n} \|A_{ij}\| \leq \max_{i}\sum_{j=1}^{n}\|A_{ij}\|$.
	Note that we have the second inequality since $\frac{\|A_{ii}\|}{|\lambda|} <1$ with the inequality (\ref{inequ}). 
	For part (ii), if $\|A_{ii}\| \geq |\lambda|$, then $|\lambda| \leq \sum_{j=1}^{n}\|A_{ij}\| \leq \max_{i}\sum_{j=1}^{n}\|A_{ij}\|$. Therefore, $|\lambda|  \leq \max_{i}\sum_{j=1}^{n}\|A_{ij}\|$ for all $\lambda \in \cup_{i=1}^{n}G_{i}$.
	
	Finally, since $\sigma(A) \subset \bigcup_{i=1}^{n}G_{i}$, we know that for all $\lambda \in \sigma(A)$, $|\lambda| \leq \max_{i}\sum_{j=1}^{n}\|A_{ij}\|$. Therefore, we have $ \|A\| \leq  \max_{i}\sum_{j=1}^{n}\|A_{ij}\|$.
\end{proof}

	\noindent
	\textbf{Proof of Theorem \ref{thm.1}.} For any $NT \times NT$ blocked matrix $M = (m_{t,s})$ where the $m_{t,s}$ is the $(t, s)$th block $N \times N$ matrix. In addition, for any $0 \leq L < T,$ we define $B_{L}(M) = [(m_{t,s})1(|t-s|\leq L)]$, which is an $NT \times NT$ matrix. Then we can write
	\begin{align*}
	\|\widehat{\Omega}-\Omega\| &\leq  \|B_{L}({\Omega})-\Omega\|+\|\widehat{\Omega}-B_{L}({\Omega})\|. 
	\end{align*}
	First, we assume that $\xi_{T}(L)=\max_{t}\sum_{|h| > L} \|Eu_tu_{t-h}'\| =o(1)$ in Assumption \ref{assum2}(ii). This implies that off-diagonal $N\times N$ blocks that are far from the diagonal block are negligible due to weak dependences. As for the first part, by Lemma \ref{lem4},
	\begin{align*}
	\|B_L(\Omega)-\Omega\| & \leq \max\limits_t\sum\limits_{s: |s-t| > L} \|Eu_tu_s'\|= \xi_{T}(L) \rightarrow 0.
	\end{align*}
	Next, note that $ f_{T}(L)=\max_{t}\sum_{|h| \leq L}\|Eu_tu_{t-h}'(1-\omega(|h|,L))\| =o(1)$ (see Assumption \ref{assum2}(iii)). Then by Lemmas \ref{lem2} and \ref{lem4}, for $C < \infty$,
	\begin{align*}
	\|\widehat{\Omega}-B_L(\Omega)\| & \leq \max_{t}\sum_{s:|t-s|\leq L}\|\widehat{\Omega}_{t,s}-Eu_tu_s'\|\\
	& \leq L\max\limits_{|h|\leq L}\|(\widetilde{\Omega}_{h}-\Omega_{h})\omega(|h|,L)\|  +  \max_{t}\sum_{|h| \leq L}\|Eu_tu_{t-h}'(1-\omega(|h|,L))\| \\
	& \leq CL\max\limits_{|h|\leq L}\|\widetilde{\Omega}_{h}-\Omega_{h}\|  +  f_{T}(L) \\
	& =  O_{P}(Lm_{N}\gamma_{T}^{1-q})+f_{T}(L),
	\end{align*}	
	where $\widetilde{\Omega}_{h}$ is defined in (\ref{e.3.8}).
	Therefore, 
	\begin{align*}
	\|\widehat{\Omega}-\Omega\| = O_{P}(Lm_{N}\gamma_{T}^{1-q}+\xi_{T}(L)+f_{T}(L)).
	\end{align*}
	We now show the second statement of Theorem \ref{thm.1}. 
	By the triangular inequality, we have
	\begin{align*}
	\|\widehat{\Omega}^{-1}-\Omega^{-1}\| &\leq \|(\widehat{\Omega}^{-1}-\Omega^{-1})(\widehat{\Omega}-\Omega)\Omega^{-1}\| + \|\Omega^{-1}(\widehat{\Omega}-\Omega)\Omega^{-1}\| \\
	&\leq \|\widehat{\Omega}^{-1}-\Omega^{-1}\| \|\Omega^{-1}\| \|\widehat{\Omega}-\Omega\| + \|\Omega^{-1}\|^{2}\|\widehat{\Omega}-\Omega\|\\
	& =  O_{P}(Lm_{N}\gamma_{T}^{1-q} +\xi_{T}(L)+f_{T}(L)) \|\widehat{\Omega}^{-1}-\Omega^{-1}\| + O_{P}(Lm_{N}\gamma_{T}^{1-q}+ \xi_{T}(L)+f_{T}(L)).
	\end{align*}
	Hence we have $(1+o_{P}(1))\|\widehat{\Omega}^{-1}-\Omega^{-1}\| = O_{P}(Lm_{N}\gamma_{T}^{1-q} + \xi_{T}(L)+f_{T}(L))$, that implies the result. $\Box$
	\\
	\\
	\textbf{Proof of Proposition \ref{proposition1}.}  First the left hand side of equation (\ref{e.3.11}) can be extended as
	\begin{align*}
	\frac{1}{\sqrt{NT}}X'(\widehat{\Omega}^{-1}-\Omega^{-1})U & = \cfrac{1}{\sqrt{NT}}X'\Omega^{-1}(\widehat{\Omega}-\Omega)\Omega^{-1}U \\
	&\;\; +\cfrac{1}{\sqrt{NT}}X'\Omega^{-1}(\widehat{\Omega}-\Omega)\Omega^{-1}(\widehat{\Omega}-\Omega)\Omega^{-1}U\\
	&\;\; +\cfrac{1}{\sqrt{NT}}X'(\widehat{\Omega}^{-1}-\Omega^{-1})(\widehat{\Omega}-\Omega)\Omega^{-1}(\widehat{\Omega}-\Omega)\Omega^{-1}U\\
	& \equiv a+b+c.
	\end{align*}
	Now we shall show that  $\frac{1}{\sqrt{NT}}X'(\widehat{\Omega}^{-1}-\Omega^{-1})U  = \frac{1}{\sqrt{NT}}X'\Omega^{-1}(\widehat{\Omega}-\Omega)\Omega^{-1}U + o_{P}(1)$. First, we define $W=\Omega^{-1}X$ and $\varepsilon = \Omega^{-1}U$. Then, $W=(w_{1}',\cdots, w_{T}')'$ with $w_{t}$ being an $N\times d$ matrix of $w_{it}$, and $\varepsilon_{it}$ is defined similarly. For any $NT \times NT$ matrix $M$, we denote $(M)_{t,s}$ or $(M)_{h}$ as the $(t,s)$th block matrix for $h=t-s$. Moreover, we denote $(M)_{ts,ij}$ or $(M)_{h,ij}$ as the $(i,j)$th element of the $(t,s)$th block matrix. Under Assumption \ref{assum3}, we  show that $b = o_{P}(1)$ as follows:
	
	We write, for $h=t-s$ and $v=k-m$,
	\begin{align*}
	b & =  \cfrac{1}{\sqrt{NT}}\sum_{t=1}^{T}\sum_{s=1}^{T}\sum_{k=1}^{T}\sum_{m=1}^{T}w_{t}'(\widehat{\Omega}-\Omega)_{t,s}(\Omega^{-1})_{s,k}(\widehat{\Omega}-\Omega)_{k,m}\varepsilon_{m}\\
	& =  \cfrac{1}{\sqrt{NT}}\sum_{t=1}^{T}\sum_{|h| \leq L}\sum_{k=1}^{T}\sum_{m=1}^{T}w_{t}'(\widehat{\Omega}-\Omega)_{h}(\Omega^{-1})_{t-h,k}(\widehat{\Omega}-\Omega)_{k,m}\varepsilon_{m}\\
	& \;\; - \cfrac{1}{\sqrt{NT}}\sum_{t=1}^{T}\sum_{|h| > L}\sum_{k=1}^{T}\sum_{m=1}^{T}w_{t}'\Omega_{h}(\Omega^{-1})_{t-h,k}(\widehat{\Omega}-\Omega)_{k,m}\varepsilon_{m}\\
	&=\cfrac{1}{\sqrt{NT}}\sum_{t=1}^{T}\sum_{|h|\leq L}\sum_{|v|\leq L}\sum_{m=1}^{T}w_{t}'(\widehat{\Omega}-\Omega)_{h}(\Omega^{-1})_{t-h,m+v}(\widehat{\Omega}-\Omega)_{v}\varepsilon_{m}\\
	& \;\; -\cfrac{1}{\sqrt{NT}}\sum_{t=1}^{T}\sum_{|h|\leq L}\sum_{|v|>L}\sum_{m=1}^{T}w_{t}'(\widehat{\Omega}-\Omega)_{h}(\Omega^{-1})_{t-h,m+v}\Omega_{v}\varepsilon_{m}\\
	& \;\; -\cfrac{1}{\sqrt{NT}}\sum_{t=1}^{T}\sum_{|h|>L}\sum_{k=1}^{T}\sum_{m=1}^{T}w_{t}'\Omega_{h}(\Omega^{-1})_{t-h,k}(\widehat{\Omega}-\Omega)_{k,m}\varepsilon_{m}\\
	&\equiv b_1+b_2+b_3.
	\end{align*}
	First, define $Q_{imp}^{hv} = \sum_{t=1}^{T}w_{it}\sum_{j=1}^{N}(\Omega^{-1})_{t-h,m+v,jp}$ as in Lemma \ref{lem3}. We have, by Lemmas \ref{lem2}-\ref{lem3} and Assumption \ref{assum3},
	\begin{align*}
	\| b_1 \| &= \left\|\frac{1}{\sqrt{NT}}\sum_{|h|\leq L}\sum_{|v|\leq L}\sum_{i=1}^{N}\sum_{j=1}^{N}\sum_{p=1}^{N}\sum_{q=1}^{N}(\widehat{\Omega}-\Omega)_{h,ij}(\widehat{\Omega}-\Omega)_{v,pq}\sum_{t=1}^{T}\sum_{m=1}^{T}w_{it}\varepsilon_{qm}(\Omega^{-1})_{t-h,m+v,jp} \right \| \\
	&\leq \frac{1}{\sqrt{NT}}\left[\max_{|h|\leq L}\max_{j}\sum_{i=1}^{N}\left|(\widehat{\Omega}-\Omega)_{h,ij}\right|\right]\left[\max_{|v|\leq L}\max_{q}\sum_{p=1}^{N}\left|(\widehat{\Omega}-\Omega)_{v,pq}\right|\right]\\
	& \;\;\;\;\; \times \max_{i}\max_{p}\left\|\sum_{|h|\leq L}\sum_{|v| \leq L}\sum_{j=1}^{N}\sum_{q=1}^{N}\sum_{t=1}^{T}\sum_{m=1}^{T}w_{it}\varepsilon_{qm}(\Omega^{-1})_{t-h,m+v,jp}\right\|\\
	&\leq \frac{1}{\sqrt{NT}}\max_{|h|\leq L}\|\widehat{\Omega}_{h}-\Omega_{h}\|_{1}^2 \max_{i,p}\left\|\sum_{|h|\leq L}\sum_{|v|\leq L}\sum_{q=1}^{N}\sum_{m=1}^{T}\varepsilon_{qm}Q_{imp}^{hv}\right\|\\
	&\leq O(L^2\sqrt{NT})\max_{|h| \leq L}\|\widehat{\Omega}_{h}-\Omega_{h}\|_{1}^2 \max_{i,p,h,v}\left\|\frac{1}{NT}\sum_{q=1}^{N}\sum_{m=1}^{T}\varepsilon_{qm}Q_{imp}^{hv}\right\|\\
	& = O_{P}(\sqrt{T}L^{2}m_{N}^{2}\gamma_{T}^{3-2q}) = o_{P}(1).
	\end{align*}
	In addition, by Lemma \ref{lem2} and Bernstein inequality,
	\begin{align*}
	\|b_2 \| & = \left \| \cfrac{1}{\sqrt{NT}}\sum_{t=1}^{T}\sum_{|h|\leq L}\sum_{|v|>L}\sum_{m=1}^{T}\sum_{i=1}^{N}\sum_{j=1}^{N}\sum_{p=1}^{N}\sum_{q=1}^{N}w_{it}(\widehat{\Omega}-\Omega)_{h,ij}(\Omega^{-1})_{t-h,m+v,jp}\Omega_{v,pq}\varepsilon_{qm} \right\| \\
	& \leq \cfrac{L}{\sqrt{NT}}\left[\max_{|h|\leq L}\max_{j}\sum_{i=1}^{N}\left|(\widehat{\Omega}-\Omega)_{h,ij}\right|\right]\left[\max_{q}\sum_{p=1}^{N}\sum_{|v|>L}\left|\Omega_{v,pq}\right|\right] \\
	& \;\;\;\;\; \times \max_{i}\sum_{t=1}^{T}\|w_{it}\| \max_{m}\sum_{q=1}^{N} |\varepsilon_{qm}| \max_{t,h,v,p}\sum_{m=1}^{T}\sum_{j=1}^{N}|(\Omega^{-1})_{t-h,m+v,jp}|\\
	& \leq O(L\sqrt{NT})\max_{|h|\leq L}\|\widehat{\Omega}_{h}-\Omega_{h}\|_{1} \sum_{|v|>L}\|\Omega_{v}\|_{1}  \\
	&= O_{P}(L^{1-\alpha}\sqrt{NT}m_{N}\gamma_{T}^{1-q}) = o_{P}(1)
	\end{align*}
	and
	\begin{align*}
	\| b_3 \| & = \left\| \cfrac{1}{\sqrt{NT}}\sum_{t=1}^{T}\sum_{m=1}^{T}\sum_{|h|>L}\sum_{k=1}^{T}\sum_{i=1}^{N}\sum_{j=1}^{N}\sum_{p=1}^{N}\sum_{q=1}^{N}w_{it}\Omega_{h,ij}(\Omega^{-1})_{t-h,k,jp}(\widehat{\Omega}-\Omega)_{km,pq}\varepsilon_{qm} \right \| \\
	& \leq \cfrac{T}{\sqrt{NT}}\left[\max_{|v|\leq L}\max_{q}\sum_{p=1}^{N}\left|(\widehat{\Omega}-\Omega)_{v,pq}\right|\right]\left[\max_{j}\sum_{i=1}^{N}\sum_{|h|>L}\left|\Omega_{h,ij}\right|\right] \\
	& \;\;\;\;\; \times \max_{i}\sum_{t=1}^{T}\|w_{it}\| \max_{m}\sum_{q=1}^{N} |\varepsilon_{qm}| \max_{t,h,v,p}\sum_{m=1}^{T}\sum_{j=1}^{N}|(\Omega^{-1})_{t-h,m+v,jp}|\\
	& \leq O(T\sqrt{NT})\max_{v}\|\widehat{\Omega}_{v}-\Omega_{v}\|_{1} \sum_{|h|>L}\|\Omega_{h}\|_{1}\\
	& =O_{P}(L^{-\alpha}T\sqrt{NT}m_{N}\gamma_{T}^{1-q}) = o_{P}(1).
	\end{align*}
	Therefore, we have $\| b \| = o_{P}(1)$. 
	Next, we define $\gamma^{*} = Lm_{N}\gamma_{T}^{1-q}+\xi_{T}(L)+f_{T}(L)$. By Theorem \ref{thm.1}, $\|\widehat{\Omega}-\Omega\| = O_{P}(\gamma^{*})= \|\widehat{\Omega}^{-1}-\Omega^{-1}\|$. Then, under the Assumption \ref{assum3}(v), when $\|\Omega^{-1}\|_{1} = O(1)$, we have
	\begin{align*} 
	\|c\| & =  \left\|\frac{1}{\sqrt{NT}}X'(\widehat{\Omega}^{-1}-\Omega^{-1})(\widehat{\Omega}-\Omega)\Omega^{-1}(\widehat{\Omega}-\Omega)\Omega^{-1}U\right\| \\~
	& \leq \|\widehat{\Omega}^{-1}-\Omega^{-1}\|\|\widehat{\Omega}-\Omega\|^{2} \sqrt{NT} = O_{P}(\sqrt{NT}\gamma^{*3})=o_{P}(1).
	\end{align*}
	Therefore, we have $\frac{1}{\sqrt{NT}}X'(\widehat{\Omega}^{-1}-\Omega^{-1})U  = \frac{1}{\sqrt{NT}}X'\Omega^{-1}(\widehat{\Omega}-\Omega)\Omega^{-1}U + o_{P}(1)$.\\
	
	From Theorem \ref{thm.1}, it is easy to that $\frac{1}{NT}X'\widehat{\Omega}^{-1}X = \frac{1}{NT}X'\Omega^{-1}X+o_{P}(1)$. Also, by the weak law of large numbers, $(\frac{1}{NT}X'\Omega^{-1}X)^{-1} = \Gamma^{-1}+o_{P}(1)$, where $\Gamma = E(\frac{1}{NT}X'\Omega^{-1}X)$. Then
	\begin{align*}
	\sqrt{NT}(\widehat{\beta}_{FGLS}-\beta) & = \left(\frac{1}{NT}X'\Omega^{-1}X\right)^{-1}\left(\frac{1}{\sqrt{NT}}X'\widehat{\Omega}^{-1}U\right)+o_{P}(1)\\
	& = \left(\frac{1}{NT}X'\Omega^{-1}X\right)^{-1}\left(\frac{1}{\sqrt{NT}}X'\Omega^{-1}U+\frac{1}{\sqrt{NT}}X'(\widehat{\Omega}^{-1}-\Omega^{-1})U\right)+o_{P}(1)\\
	& = \Gamma^{-1}\left(\frac{1}{\sqrt{NT}}X'\Omega^{-1}U+\frac{1}{\sqrt{NT}}X'(\widehat{\Omega}^{-1}-\Omega^{-1})U\right)+o_{P}(1)\\
	& = \Gamma^{-1}\left(\frac{1}{\sqrt{NT}}X'\Omega^{-1}U\right)+ \Gamma^{-1}\left(\frac{1}{\sqrt{NT}}X'\Omega^{-1}(\widehat{\Omega}-\Omega)\Omega^{-1}U\right)+o_{P}(1). \;\; \Box
	\end{align*}
	\\
	\\
	\textbf{Proof of Theorem \ref{thm.2}.}  It suffices to prove $\left\| \frac{1}{\sqrt{NT}}X'\Omega^{-1}(\widehat{\Omega}-\Omega)\Omega^{-1}U \right\|  = o_{P}(1)$.  \\~

Let $A_{b_h} = \{(i,j) : |Eu_{it}u_{j,t-h}| \neq 0\}$. Also, let $W = \Omega^{-1}X$, and $\varepsilon = \Omega^{-1}U$. In addition, $w_{it}$ and $\varepsilon_{it}$ are defined as in the proof of Proposition \ref{proposition1}. We define $\mathbb{G}_{T,ij}^{1}(h) = \frac{1}{\sqrt{T}}\sum_{t=h+1}^{T}(u_{it}u_{j,t-h}-Eu_{it}u_{j,t-h})$ and $\mathbb{G}_{T,ij}^{2}(h) = \frac{1}{\sqrt{T}}\sum_{t=h+1}^{T}w_{it}\varepsilon_{j,t-h}$. Then under the Assumption \ref{assum4}, there is $C>0$ so that 
\begin{flalign*}
\left\| \cfrac{1}{\sqrt{NT}}X'\Omega^{-1}(\widehat{\Omega}-\Omega)\Omega^{-1}U  \right\| & = \left\| \frac{1}{\sqrt{NT}}W'(\widehat{\Omega}-\Omega)\varepsilon \right\| \\
& \leq \left\| \frac{C}{\sqrt{NT}}\sum_{h=0}^{L}\sum_{i,j \in A_{b_h}}\sum_{t=h+1}^{T}w_{it}\varepsilon_{j,t-h}\frac{1}{T}\sum_{s=h+1}^{T}(u_{is}u_{j,s-h}-Eu_{it}u_{j,t-h}) \right\| + o_{P}(1)\\
&= \left\| \frac{C}{\sqrt{NT}}\sum_{h=0}^{L}\sum_{i,j \in A_{b_h}}\mathbb{G}_{T,ij}^{1}(h)\mathbb{G}_{T,ij}^{2}(h) \right\|+ o_{P}(1) = o_{P}(1).
\end{flalign*}

	\newpage
	\bibliographystyle{economet}
	\bibliography{fgls_bib}

\end{document}